\documentclass{ws-ijmpa}
\usepackage[super,compress]{cite}
\begin{document}

\markboth{Vladimir N. Lukash and Vladimir N. Strokov}
{Space--times with integrable singularity}

%
\catchline{}{}{}{}{}
%

\title{SPACE--TIMES WITH INTEGRABLE SINGULARITY}

\author{VLADIMIR N. LUKASH}

\address{Astro Space Center of the Lebedev Physical Institute
\\ ul. Profsoyuznaya, 84/32, 117997 Moscow Russia \\
lukash@asc.rssi.ru}

\author{VLADIMIR N. STROKOV\footnote{On leave from Astro Space Center of the Lebedev
Physical Institute.}}

\address{Departamento de F{\'i}sica -- ICE -- Universidade
Federal de Juiz de Fora \\Juiz de Fora, MG, Brasil -- CEP:
36036-330\\
strokov@asc.rssi.ru}

\maketitle

\begin{history}
\received{Day Month Year}
\revised{Day Month Year}
\end{history}

\begin{abstract}
We use the phenomenological approach to study properties of
space-time in the vicinity of the Schwarzschild black-hole
singularity. Requiring finiteness of the Schwarzschild-like
metrics we come to the notion of integrable singularity that is,
in a sense, weaker than the conventional singularity and allows
the (effective) matter to pass to the white-hole region. This
leads to a possibility of generating a new universe there. Thanks
to the gravitational field of the singularity, this universe is
already born highly inflated ('singularity-induced inflation')
before the ordinary inflation starts.

\end{abstract}

\ccode{PACS numbers: 04.70.-s, 04.20.Dw, 98.80.Bp}

\section{Introduction}\label{introduction}

Since General Relativity (GR) came into being, the scientific
community has had various opportunities to verify that it gives a
viable description of phenomena that include strong gravitational
fields and relativistic velocities. Its experimental basis once
consisting of the three classical GR effects (perihelion
precession, deflection of light and red shift) has recently
acquired one of its crucial contributions -- the Cosmological
Standard Model (CSM) of the visible Universe (see
Ref.~\refcite{WMAP} and references therein). Extrapolating this
CSM to the past~\cite{lm} leads to one of the main features of GR
-- singularities.

As is long known,\cite{penrose,hawkingpenrose,hawkingellis} the
latter are a common place in general-relativistic solutions
including the most physical ones, the above-mentioned
Friedmann--Robertson--Walker universe and black holes which
probably reside in central parts of many galaxies.\cite{SMBH}
Although one can think of space--times where a singularity appears
'all of a sudden',\cite{Ellis-Schmidt, krasnikov} it is usually
understood that some curvature invariants grow infinitely as one
approaches the singular hypersurface (later on we deal with this
type of singularity). However, one expects that the Einstein
equations will somehow change near the singularity. For example,
this change may be due to quantum effects.\cite{shapiro} In
cosmology this possibility was first studied in
Refs.~\refcite{starobinskii,anderson-1,anderson-2,anderson-3}
while for the case of black holes in
Ref.~\refcite{frolovvilkovsky}.

Currently, the exact form of the quantum corrections remains
unknown and the variety of suggested modifications complicates the
studies. Under the circumstances it is natural to parametrize
properties of gravitational field near the singularity in a way
similar to how probable gravity modifications~\cite{ferreira} are
parametrized by the notions of dark matter and dark energy in
cosmology.

In this paper we phenomenologically model ultrahigh-curvature
processes near the singularity by introducing a continuous
\textit{effective} mass distribution~$m(r,t)$ such that
$m(r,t)\rightarrow 0$ as $r\rightarrow 0$. Note that we do not
discuss physical nature of the mass function~$m(r,t)$ while we do
study the consequences of the equivalent assumption that metrics
potentials become finite near $r=0$. We also assume that the mass
function $m(r,t)$ has the property that it appears to be
point-like at large distances.

It should be also emphasized that our assumption is milder than
the usual requirement that a physically admissible space--time be
completely free of singularities. As applied to the case of
spherically symmetric geometries it leads to the notion of
\textit{integrable} singularity.\cite{lms-1,lms-2} It is
\textit{still} a singularity, but a softer one. The motivation for
considering it is the following. Infinite tidal forces near
singularity are that unattractive feature that makes it
singularity in the physical sense (at least, this is true for the
singularities that are believed to be formed in nature). From this
point of view, it is quite remarkable that there may exist a
singularity with infinite tidal forces only in some directions, so
that they do not influence the material flow that supports the
geometry. In this sense, the integrable singularity allows us to
connect the interior of the black hole to a daughter universe that
is born with $\mathbb{R}\times\mathbb{S}^{2}$ symmetry. Although
transversal tidal forces are infinite, they do not influence the
flow that respects the space--time symmetry and passes the
singularity to form the new universe. This situation is
reminiscent of what happens in cosmological models with dust and
generalized Chaplygin gas.\cite{Chap-singularity}

Also, though finite, the gravitational potential is deep, which
enhances the initial volume of the new universe. Inflation as an
intermediate stage in this universe is only needed to make it
isotropic. The result suggests that the mechanism of generating
new universes inside a black hole can work without special
assumptions on finiteness of density or curvature like those used
in Refs.~\refcite{markov-1,markov-2,markovfrolov,BFrolov96} or in
bouncing models.\cite{Nov66} Also note that it is crucial for this
mechanism that the black-hole interior contain no static regions,
so we do not consider here wormholes (see, for example,
Refs.~\refcite{Poplawski,BrSt}) or solutions with an internal
$R$-region.\cite{dymnikova}

In the next two sections we introduce the definition of integrable
singularity and analyze its properties. Then in Sect.~\ref{model}
we present two toy models containing the one. The first of them
gives the answer to the question of the source of the
Schwarzschild geometry while the second is an example of a
universe emerging from the black-hole interior. Finally, in
Sect.~\ref{discussion} we discuss the effective-matter approach,
the type of integrable singularity and the suggested scenario of
generating new universes.

\section{Integrable singularity}\label{sing}

Let us consider a class of metrics that, first, have the
Schwarzschild form in the region where quantum effects are
negligible, second, inherit the global Killing $t$-vector from the
vacuum solution and, third, contain finite
quantities\,\footnote{We use the signature $(+\,-\,-\,-\,)$, the
sign
$R^{\alpha}{}_{\beta\gamma\delta}=\partial_{\gamma}\Gamma^{\alpha}_{\beta\delta}-\ldots$
and set the speed of light $c=1$.}:
\begin{equation}
\label{general-metrics}
ds^{2}=N^2(1+2\Phi)\,dt^{\,2}-\frac{dr^{\,2}}{1+2\Phi}-
r^2d\Omega^{2}\,,
\end{equation}
where $d\Omega^{2}=d\theta^{2}+\sin^{2}{\theta} d\varphi^{2}$ is
the line element on a unit 2-sphere in angular
coordinates~$(\theta,\varphi)$. More exactly, the functions $N(r)$
and $\Phi(r)$ are finite along with their two first derivatives
everywhere including $r=0$, and when $r$ exceeds some
characteristic positive value $r_{0}$
\begin{equation}\label{vacuum}
N(r\geq r_{0})=1\,,\quad \Phi(r\geq r_{0})=-\frac{GM}{r}\,,
\end{equation}
where $M$ is the external mass of the black hole and $G$ the
gravitational constant. The well-known coordinate divergence at
the horizon $\Phi=-1/2$ which separates the $R$- ($\Phi>-1/2$) and
$T$-regions ($\Phi<-1/2$) is treated in the usual
way.\cite{frolov-novikov}

The surface $r=r_{0}$ is a sort of border where quantum processes
start playing a significant role as we go inside the black hole.
We will assume that $M>>M_{P}\sim 10^{19}$~GeV which ensures that
the Hawking radiation does not affect our further consideration
and that the surface $r=r_{0}$ resides under the horizon and,
hence, is spacelike. Since we required the finiteness of the
metric potentials even at $r=0$, it is only logical that the
r.h.s. of the Einstein equations $G_{\mu\nu}=8\pi GT_{\mu\nu}$ is
non-zero somewhere in $r<r_{0}$. To model the situation we
consider the effective energy-momentum tensor (EMT) that emerges
in a triggered\,\footnote{\label{Note2} Physically, it may be
thought of as a phase transition caused by quantum-gravitational
processes of vacuum polarization and matter creation in intensive
variable gravitational field (though such an interpretation is
valid only in the quasiclassical limit). One of the candidates to
play the role of the 'temperature' is the squared Riemann tensor
$\mathcal{I}=R_{\alpha\beta\gamma\delta}R^{\alpha\beta\gamma\delta}$
which in the Schwarzschild solution equals to $\,48(GM/r^3)^2$.
Then the characteristic value $r_{0}$ can be estimated from the
physical dimension: $\mathcal{I}\sim 1/l_{P}^{4}$, where
$l_{P}\sim 10^{-33}$~cm is the Planck length. To give an example,
for a black hole with mass of the order of that of the Sun
$r_{0}\sim 10^{12}l_{P}$.} way:
\begin{eqnarray}
T_{\mu\nu}&\ne& 0 \quad \mbox{when} \quad 0\le r\le r_0<2GM\,,\nonumber \\
T_{\mu\nu}&=&0 \quad \mbox{when} \quad r > r_0\,.
\end{eqnarray}

Let us find out which properties of the effective matter support
the general spherical space--time when the functions $N$ and
$\Phi$ depend on both coordinates $r$ and $t$. Making use of eq.
(\ref{general-metrics}) one obtains from the GR
equations:\cite{landafshits}
\begin{equation}
\label{rm} \Phi = -\frac{G m}{r}\,,
\end{equation}
where the finite {\it mass function}
\begin{equation}
\label{m} m=m\!\left(r,t\right)= 4\pi\!\int_{0} T_t^t r^2 dr= m_0
-4\pi r^2\!\!\int\! T_t^r dt
\end{equation}
vanishes at $r=0$ (m(0,t) = 0)\, thanks to the finiteness
condition applied on the potential $\Phi$. The function
$m_0=m_0(r)$ is determined by initial/boundary conditions. We also
conclude that for the function $m(r,t)$ to be finite, $T_t^t r^2$
must be integrable at $r=0$. Hereafter, by definition, the
distribution of longitudinal pressure/energy density (in
$T/R$-regions, respectively) over~$r$ integrable in this sense is
called to possess an {\it integrable}
singularity.\cite{lms-1,lms-2}

In our case the l.h.s. of Eq.~(\ref{rm}) depends only on $r$.
Hence, $T^{r}_{t}=0$ and the EMT compatible with
metrics~(\ref{general-metrics}) is
\begin{equation}
T^{\mu}_{\nu}=diag(-p,\varepsilon,-p_{\perp},-p_{\perp}).
\end{equation}
Note that, since the EMT emerges only under the horizon, the
$({}^{r}_{r})$-component plays the role of energy density, because
the coordinate $r$ plays that of time.\cite{novikov} The rest of
the Einstein equations read:
\begin{eqnarray}
\frac{N^\prime}{N} &=&\frac{4\pi
Gr^2\!\left(\varepsilon+p\right)}{2Gm-r}\,, \quad m(r)=-4\pi\!\int_0 p\,r^2 dr\, \label{N}\\
p_{\perp}&=&\frac{N^\prime}{2N}\!\left(\frac{m}{4\pi
r^2}-r\varepsilon\right)
-\frac{\left(r^2\varepsilon\right)^\prime}{2r}\,, \label{pressure}
\end{eqnarray}
where the prime stands for the derivative with respect to $r$.
Extending the integral in Eq.~(\ref{N}) up to $r_0$ one obtains
the external mass
\begin{equation}
\label{m1} M=-4\pi\!\int_0^{r_0}\!p\,r^2 dr\,.
\end{equation}

In order to integrate these equations it is necessary to specify
the effective matter Lagrangian or its equation-of-state. Of
course, their exact form can be found only in the complete theory
of quantum gravity. However, in the phenomenological framework we
can depict some generic features of the space--times under
consideration already now.

\section{Some general properties}\label{properties}

\begin{theorem}
Let the potential $\Phi(r=0)\equiv\Phi_{0}$ and its second
derivative be finite at $r=0$, and its first derivative vanish no
slower than $r$ as $r\rightarrow 0$. Then tidal forces remain
finite on world line of the matter flow and
\begin{equation}
\label{divergence}
\begin{array}{c}
p=\displaystyle\frac{\Phi_{0}}{4\pi Gr^{2}}\qquad\mbox{as}\quad
r\rightarrow 0\,.
\end{array}
\end{equation}
\end{theorem}

\begin{proof}
In order to study the story of an extended body falling to
$r=0$ we follow Ref.~\refcite{MTW}. Non-vanishing components of
the Riemann tensor in a locally inertial reference
frame~$(\hat{t},\hat{r},\hat{\theta},\hat{\varphi})$ are given by
the formulae:
\begin{equation}\label{curvature}
R_{\hat{t}\hat{r}\hat{t}\hat{r}}=\Phi''\,,\qquad
R_{\hat{t}\hat{\theta}\hat{t}\hat{\theta}}=R_{\hat{t}\hat{\varphi}\hat{t}\hat{\varphi}}=\frac{\Phi'}{r}\,,
\end{equation}
\begin{equation}\label{curvature-1}
R_{\hat{\theta}\hat{\varphi}\hat{\theta}\hat{\varphi}}=-\frac{2\Phi}{r^{2}}\,,\qquad
R_{\hat{r}\hat{\theta}\hat{r}\hat{\theta}}=R_{\hat{r}\hat{\varphi}\hat{r}\hat{\varphi}}=-\frac{\Phi'}{r}\,.
\end{equation}
The only non-zero component of the 4-velocity of the matter is
$u^{r}$. Therefore, the equation, which governs how fast two
particles separated by the spatial vector $\xi^{\hat{a}}$,
$\hat{a}=(\hat{t},\hat{\theta},\hat{\varphi})$, accelerate
relative to each other, reads:
\begin{equation}\label{tides}
\frac{D^{2}\xi^{\hat{a}}}{d\hat{r}^{2}}=-R_{\hat{r}\hat{a}\hat{r}\hat{b}}\xi^{\hat{b}}.
\end{equation}
As one can see, under conditions of the theorem the right-hand
side is finite. Moreover, since
$R_{\hat{\theta}\hat{\varphi}\hat{\theta}\hat{\varphi}}$ is the
only component in~(\ref{curvature}), (\ref{curvature-1}) that
diverges as $r\rightarrow 0$, the tidal forces remains finite on
any radial geodesic~$u^{\mu}=(u^{t},u^{r},0,0)$.

Also, Eq.~(\ref{rm}) yields
$$
\Phi'r+\Phi=4\pi Gpr^2.
$$
Hence, $p=\Phi_{0}/4\pi Gr^{2}+O(1)$ as $r\rightarrow 0$. Since
$p$ is supposed to comprise every correction to the Einstein
tensor (see Introduction), this implies that all these corrections
diverge as $r\rightarrow 0$.
\end{proof}

In what follows it is convenient to extend the values of $r$ to
the real domain~$(-\infty,+\infty)$. This choice allows one to
describe the black- and white-hole regions in a unified manner.
Indeed, a physical way to probe a space--time is to study test
particles' trajectories. For example, one can mark the surfaces
$r=const$ by launching a particle along the radius
($d\theta=d\varphi=0$) and measuring its consecutive positions. As
usual, in the $T$-region of the black hole the particle moves
towards smaller values of $r$. Also it has no problem passing the
hypersurface $r=0$. Hence, as the particle passes to the
$T$-region of the white hole, we can consider it moving on to
smaller, but now negative, values of $r$. Another advantage of
this choice is that it fits the fact that $r$ as time (as it is in
the $T$-regions) must 'tick' in one direction.

We can now classify all possible models by their properties under
the time inversion $r\rightarrow -r$ with respect to $r=0$ into
reversible and irreversible ones. Reversible
solutions~(\ref{general-metrics}) are given by even functions of
$r$ whereas the second-type models are non-invariant under the
inversion and described by asymmetric profiles. From geometrical
point of view, models of different types depend differently on
intrinsic and extrinsic curvature of the Schwarzschild-like
space--time (see Appendix). From physical point of view, the
mechanism that underlies the reversible solutions is the vacuum
polarization as described by the quantum corrections (see
Introduction). In the irreversible regime some transformation of
gravitational degrees of freedom into material ones must take
place.

We are about to give a toy model of each type. As soon as the
external mass $M$ is positive, the longitudinal pressure is, on
average, negative (see Eq.~(\ref{m1})), and for simplicity we
choose it to be vacuum-like, $p=-\varepsilon$. Then $N=1$ and the
energy density is found from Eq.~(\ref{pressure}) which becomes
\begin{equation}
\label{tB} \frac{d\left(\varepsilon r^2\right)}{rdr}=-2p_\perp\,.
\end{equation}

\section{Two toy models}\label{model}

\subsection{Black-white hole}\label{black-white}

Let us consider a symmetric step for the profile of the transverse
pressure~$p_{\perp}$:
\begin{equation}
p_\perp^{\left(A\right)}= p_0\cdot\theta\left(r_0^2-r^2\right).
\end{equation}
Then $M\equiv 8\pi r_0^3p_0/3$ is the black-hole mass, and
integrating Eq.~(\ref{tB}) with the initial condition $\varepsilon
(r\ge r_0)=0$ yields
\begin{equation}\label{epA}
\varepsilon^{\left(A\right)}=p_\perp^{\left(A\right)}\cdot\left(\frac{r_0^2}{r^2}-1\right),
\end{equation}

and the potential~(\ref{rm})
\begin{equation}
\Phi^{(A)}=\left\{
\begin{array}{lc}
-\displaystyle\frac{3GM}{2r_0}\left(1-\frac{1}{3}\left(\frac{r}{r_0}\right)^{2}\right)\,,
& |r|\leq r_0\,,\\
-\displaystyle\frac{GM}{|r|}\,, & |r|\geq r_0\,.
\end{array}
\right.
\end{equation}

A model of this type can be referred to as a 'black-white hole',
because the twin white hole is linked to its black counterpart by
the matter flow. The spacetime geometry is seemingly sourced by
the effective EMT. However, this EMT {\it alone} cannot sustain
the black hole, because the latter resides in the absolute past
with respect to the former. The actual source is located on the
border between another black-and-white pair, which lies under the
point of intersection of the horizons $r=2GM$ (see
Fig.~\ref{eternal}). In other words, the Penrose diagram of the
black-white hole is an infinite chain of elementary Penrose
diagrams engaged by material regions that source the geometry. Let
us emphasize that, presumably (see footnote~\ref{Note2}), $\Phi_0
= \min\Phi^{(A)}= -\frac{3GM}{2r_0}<< -1$.

This also helps to clarify the question of the source of the
Schwarzschild space--time. Indeed, in the limit
$r_0\stackrel{M=const}{\longrightarrow} 0$ we obtain a black
\textit{or} white hole maximally extended onto the empty space
with a delta-like source localized at $r=0$:
\begin{equation}
\label{d}
\varepsilon=-p=2p_\perp=M\,\frac{\delta\left(r\right)}{2\pi
r^2}\,,
\end{equation}
where $\delta(r)=\theta'(r)$ is one-dimensional delta-function.

\begin{figure}
\includegraphics[width=\columnwidth]{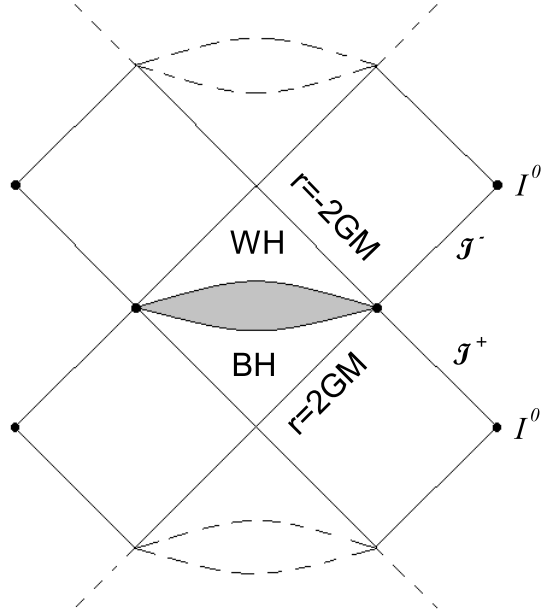}
\caption{Penrose diagram of a black-white hole with the shaded
region depicting the matter that links the respective $T$-regions.}%
\label{eternal}
\end{figure}

\begin{figure}
\includegraphics[width=\columnwidth]{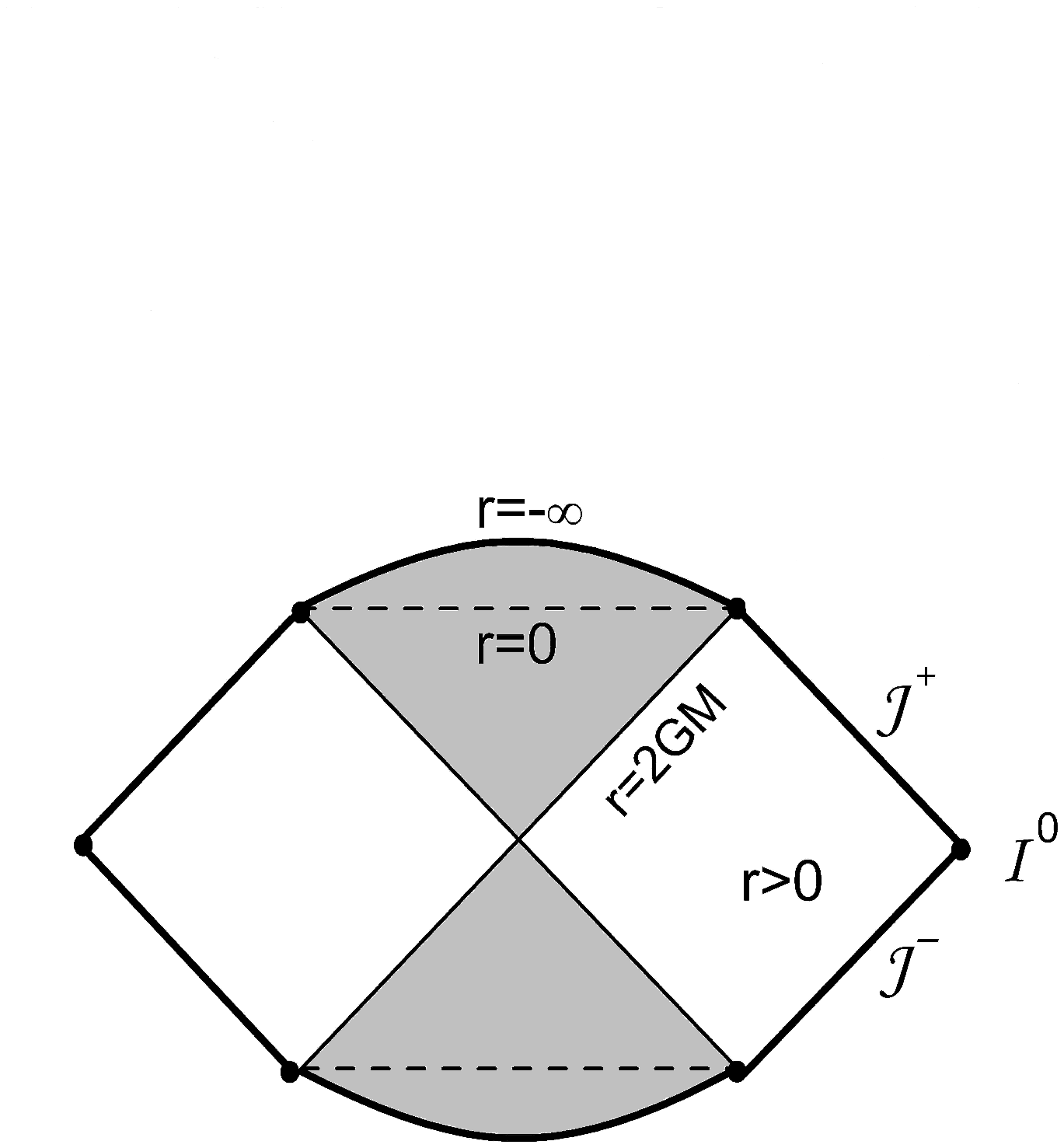}
\caption{Penrose diagram of an asymptotically de Sitter astrogenic universe.}%
\label{astrogenic}
\end{figure}

\subsection{Cosmology-like solution}\label{cosmology}

For this example we consider an asymmetric step for the transverse
pressure profile:
\begin{equation}
p_\perp^{\left(B\right)}=p_0\cdot\theta\left(rr_0-r^{2}\right)-
p_1\cdot\theta\left(-r\right),
\end{equation}
where $p_1$ is a positive constant. Analogously to the previous
section we obtain:
\begin{equation}\label{epB}
\varepsilon^{\left(B\right)}=-p_\perp^{\left(B\right)}+p_0\frac{r_0^2}{r^2}\cdot\theta\left(r_0-r\right).
\end{equation}

\begin{equation}
\label{potential-B} \Phi^{(B)}=\left\{
\begin{array}{lc}
-\displaystyle\frac{3GM}{2r_0}\left(1+\frac{p_1}{3p_0}\left(\frac{r}{r_0}\right)^{2}\right)\,,
& r\leq 0\,,\\
\Phi^{(A)}\,, & r\geq 0\,.
\end{array}
\right.
\end{equation}

The nature of this solution becomes clearer if we introduce the
proper time $d\tau=-|1+2\Phi|^{-1/2}dr$. Then for $\tau\in
(0,\infty)$ Eqs.~(\ref{epB})(\ref{potential-B}) yield the solution
asymptotically approaching the de Sitter (see
Fig.~\ref{astrogenic}):
\begin{equation}
\label{au} r = - g_0 \frac{\sinh(H_1\tau)}{H_1} \,,\qquad
\varepsilon =\frac{3H_1^2}{8 \pi G} \left(
1+\frac{1+g_0^{-2}}{3\sinh^2(H_1\tau)} \right),
\end{equation}
\[
ds^2=d\tau^2- g_0^2 \left(\cosh^2\!\left(H_1\tau\right)dt^2+
\frac{\sinh^2\!\left(H_1\tau\right)}{H_1^2}d\Omega^{2}\right),
\]
where $g_0^{2}\equiv-2\Phi_0-1=(\frac{3GM}{r_0}-1) > 1/2$ and the
constant $H_1= (8\pi Gp_1/3)^{1/2}$ acquires any value independent
of the external mass of the black hole. Therefore, the $T$-region
of the white hole may, in principle, incorporate an entire
universe. Such a universe originating from inside a black hole can
be dubbed \textit{astrogenic}. The space of the created universe
is homogeneous and huge owing to two parameters from
Eq.~(\ref{au}): the large prefactor $g_0\sim \sqrt{GM/r_0}\sim
r_0/l_P\sim (M/M_P)^{1/3}\gg 1$ and a constant $H_1$ responsible
for inflation. While the former stems from the violent matter
production in the deep gravitational potential of integrable
singularity the latter is ordinary space inflation due to phase
transition in already expanding effective matter. The toy example
also shows that although the astrogenic universe is anisotropic in
the beginning (with the symmetry of the Kantowski-Sachs model), it
may become isotropic at late times:
$-p^{\left(B\right)}=\varepsilon^{\left(B\right)}\rightarrow
p_1=-p_{\perp}^{\left(B\right)}$ as $\tau\rightarrow+\infty$.

\section{Discussion}\label{discussion}

Until the theory of quantum gravity is established the
phenomenological effective-matter approach can be very useful in
the study of singularities. One, however, has to be careful of how
to apply it. This we can learn from the analogy with classical
hydrodynamics. Long before a theory of microscopic motion was
developed many paradoxes of the ideal fluid hydrodynamics had been
known to be solved by introducing a coefficient of
viscosity.\cite{landafshits-6} Note that in the macroscopic theory
the value of the coefficient is found by measuring the fluid
motions, that is, inferred from the l.h.s. of the Navier--Stokes
equations. The same scheme should be adopted for the effective
Einstein equations. Indeed, results obtained so far in the field
of gravity seem to be more instructive about the properties of
space--times rather than those of effective matter.

Since it is little known about the effective equation-of-state, it
is tempting to consider cases with high symmetry like that of the
de Sitter vacuum.\cite{Farhi87} The latter has become so popular
in the literature that references are too numerous to cite.
However, the above considerations indicate that in the black-hole
singularity problem it is reasonable to require the
$\mathbb{R}\times\mathbb{S}^{2}$ symmetry of the EMT. Another
argument follows from the well-known BLK regime\cite{BLK} that
predicts a break of isotropy in the vacuum solution. Although it
was argued\cite{poisson-israel-1} that quantum effects may dump
the anisotropy, until the exact physics under the Planck scale is
known the anisotropic choice is nothing less natural than the de
Sitter.

Moreover, the second model from the preceding section offers an
advantage. It demonstrates that a newly born universe can be
connected to the black-hole interior without assuming density
and/or curvature being finite. Here it is worth clarifying what
'connected' means, which brings us to the question how to classify
the integrable singularity. According to Ref.~\refcite{tipler}, it
is a strong singularity. Indeed, writing
metrics~(\ref{general-metrics}) locally as
$ds^{2}=\delta\tau^{2}-\delta x^{2}-\delta y^{2}-\delta z^{2}$ we
find that the comoving volume built on three spacelike vectors in
the $T$-region (see Eq.~(\ref{ell}))
\begin{eqnarray}\label{volume}
\delta V&=&\delta x\delta y\delta z=\\ \nonumber
{}&=&N|1+2\Phi|^{1/2}r^{2}\det{\{\omega_{ij}\}}dtdy^{1}dy^{2}\rightarrow
0 \quad \mbox{as} \quad r\rightarrow 0.
\end{eqnarray}
Therefore, on one hand, an extended body falling towards the
singularity will be destroyed by tidal forces. But notice that the
volume~(\ref{volume}) tends to zero because of the infinitesimal
radius of the 2-sphere. It means that it is basically the
\textit{transverse} tidal forces that destroy the body whereas the
longitudinal forces stay finite. This fact also follows from
Eqs.~(\ref{curvature})-(\ref{tides}), because in the extended body
there always exist points with non-zero angular components of
4-velocity. However, it is more important that the flow that
respects the space-time symmetry does not suffer from those forces
and it is this flow that forms the new universe after the
integrable singularity.

Here we would like to make a few comments on the suggested
scenario of generating new universes. First, the model proposed is
based on the notion of eternal black hole. However, its further
generalization to the case of a black hole originating from
collapse does not offer principal difficulties.\cite{lms-2}
Second, there is no need of generalizing it to the cases of
rotating or charged black holes. This is thanks to the 'mass
inflation' phenomenon: as it was shown in
Ref.~\refcite{poisson-israel}, an observer in the interior static
region of a rotating (or charged) black hole would measure
internal mass the factor $\sim\sqrt{M/M_{P}}$ larger than $M$.
Because of the internal mass this large, the geometry of the
interior would be similar to the Schwarzschild one. Third, thanks
to the gain factor $g_{0}>>1$ in Eq.~(\ref{au}) the volume of the
new universe is already large when $\tau\sim H_{1}^{-1}$, i.e.
before the ordinary inflation enhances it. This effect can be
referred to as 'singularity-induced inflation', because this kind
of inflation does not require any specific field and is of pure
gravitational nature. Finally, residual cylindrical anisotropy in
the present-day data would indicate the astrogenic origin of our
universe. Although some authors claim to have discovered the
global anisotropy (see, for example,
Refs.~\refcite{komberg-1,komberg-2}), the current precision of
cosmological observations is insufficient to state that, and
future observations should clarify the situation.

It is also necessary to point out limitations of the
phenomenological approach. The main of them is that there is
ambiguity in determining the effective equation-of-state. Speaking
of the longitudinal pressure, the vacuum-like equation-of-state
used above seems to be fairly common, because it prevents
intersection of layers in the more realistic case of
collapse.\cite{lms-2} As for the other pressure component, behind
the phase transition scenario we have drafted (see
footnote~\ref{Note2}), of course, there must be a specific model
to be developed. Also, processes like passing of an elementary
particle through the integrable singularity are beyond the scope
of the effective-matter approach. Nevertheless, within its
framework we were able think of a situation where a singularity
does not 'spoil' the space--time and, more than that, leads to a
natural cosmogenesis scenario. The very possibility of it
indicates that the requirement that a physical space--time must be
completely free of singularities in the mathematical sense may be
too strong indeed. It may be a pity that a star voyager cannot
travel to the new universe, but it is still not the reason to say
that nature is afraid of singularities.

\section*{Acknowledgments}

The authors thank E.V.~Mikheeva and I.D.~Novikov for valuable
comments. We are also grateful to I.L.~Shapiro for critical
remarks. This work is supported by Federal Programm ``Scietific
Personel'' (nn. 16.740.11.0460 and 8422), by the Russian
Foundation for Basic Research (OFI 11-02-12168 and 12-02-00276)
and by grant no. NSh 2915.2012.2 from the President of Russia.
V.N.S. is grateful to FAPEMIG for support.

\appendix

\section{Intrinsic and extrinsic curvature of Schwarzschild-type space--times}

In the spirit of the ADM formalism\cite{MTW} the
Schwarzschild-type geometry~(\ref{general-metrics}) can be
represented as a (2+2)-split. To begin with, remind that the
quantity~$r$ plays two roles. On one hand, it is the curvature
radius of a 2-dimensional sphere. On the other hand, it is one of
the coordinates. In order to discriminate one from the other, let
us present the metrics of an arbitrary spherically symmetric 4D
space--time split into a pair of 2D spaces:\cite{polishchuk-2,GHP}
\begin{equation}
\label{tau} dX^{\,2}= n_{IJ}\,dx^I dx^J
\end{equation}
and
\begin{equation}
\label{ell} dY^{\,2}=\gamma_{ij}\,dy^i dy^j \equiv
r^2\omega_{ij}\,dy^i dy^j\,,
\end{equation}
where the functions $n_{IJ}$ and $r$ depend on the variables
$x^I\!=(x^1,x^2)\in \mathbb{R}^2$ and are {\it independent} of the
internal 2D coordinates $y^i$ of the closed homogeneous and
isotropic 2-surface $\mathbb{S}^2$ of the unit curvature
$d\Omega^{2}\!= \omega_{ij}dy^i dy^j$, $\gamma_{ij}\!\equiv
r^2\omega_{ij}$. If the coordinates are chosen to be angular,
$y^i\!=(\theta,\varphi)$, we have $\omega_{ij}\!=
diag{\,(1,\,\sin^2\theta)}\,$ with $\,\theta\in [0, \pi]$ and
$\varphi\in [0, 2\pi)$.

By choosing the four coordinates of the covering grid one can turn
four non-diagonal components of the full metric tensor into zero
$g_{Ii}\!=0$ and write it in the orthogonal reference frame
$x^\mu\!=(x^I,y^i)$:
\begin{equation}
\label{s} ds^{\,2}= g_{\mu\nu}\,dx^\mu dx^\nu = dX^{\,2}-dY^{\,2},
\end{equation}
where $g_{\mu\nu}= diag{\,(n_{IJ}, -\gamma_{ij})}$ is the metrics
of the spherically symmetric geometry in the orthogonal split
$2+2$\,. The energy--momentum tensor corresponding to (\ref{s}) is
$T_{\mu\nu}\!=diag{\,(T_{IJ},p_\perp\gamma_{ij})}\,$, \,where
$p_\perp$ is the transversal pressure.

At this point the metric potential $r$ in Eq.~(\ref{ell}) can be
introduced in the invariant manner as radius of the {\it
intrinsic} curvature $\rho$ of the closed $Y$-space, where
\begin{equation}
\label{R2} R_{ij}^{(Y)}=\rho\gamma_{ij}\quad\; \mbox{and} \;\quad
\rho\equiv\frac 12 R^{(Y)}\!= r^{-2}
\end{equation}
are the Ricci tensor and scalar constructed from the metrics
$\gamma_{ij}$. By definition, the 2-space $\gamma_{ij}$ and its
intrinsic curvature $\rho$ are invariant with respect to
interchanging $r$ and $-r$ while the {\it extrinsic} curvature of
$Y$ depends on the sign of $r$ and determines the orientation and
evolution of the surface~$\mathbb{S}^{2}$ in the space--time
(\ref{s}):
\begin{equation}
\label{K} \mathcal{K}_{ijI}\equiv\frac 12 \gamma_{ij,I}=
\mathcal{K}_I \gamma_{ij}\,,\quad\mathcal{K}_I\equiv\frac 12
\gamma^{ij}\mathcal{K}_{ijI} =\frac{r_{,I}}{r}\,,
\end{equation}
where the comma in the subscript stands for the partial derivative
with respect to $x^I$. This fact becomes obvious in the
coordinates where one of the variables~$x^I$ is identically equal
to $r$.

\end{document}